\documentclass[journal ]{new-aiaa}
\usepackage[utf8]{inputenc}
\usepackage{multicol}
\usepackage{textcomp}
\usepackage{graphicx}
\usepackage{amsmath}

\usepackage{amsthm}
\usepackage{subfigure}
\usepackage[version=4]{mhchem}
\usepackage{siunitx}
\usepackage{algorithm,algorithmic}
\usepackage{longtable,tabularx}
\newtheorem{theorem}{Theorem}[section]

\newtheorem{lemma}[theorem]{Lemma}
\newtheorem{assumption}{Assumption}
\newtheorem{definition}[theorem]{Definition}
\title{Lattice piecewise affine approximation of explicit model predictive control with application to satellite attitude control}

\author{Zhengqi Xu\footnote{School of Mechanical Engineering and Automation; 23s153058@stu.hit.edu.cn.}, Jun Xu\footnote{Associate Professor, School of Mechanical Engineering and Automation; xujunqgy@hit.edu.cn.} and Ai-Guo Wu\footnote{Full Professor, School of Mechanical Engineering and Automation; agwu@hit.edu.cn.}}
\affil{Harbin Institute of Technology, Shenzhen 518055, China}
\author{Shuning Wang\footnote{Full Professor, Department of Automation; swang@tsinghua.edu.cn.}}
\affil{Tsinghua University, Beijing 100084, China}
 
\begin{document}
\maketitle
\begin{abstract}
Satellite attitude cotrol is a crucial part of aerospace technology, and model predictive control(MPC) is one of the most promising controllers in this area, which will be less effective if real-time online optimization can not be achieved. Explicit MPC converts the online calculation into a table lookup process, however the solution is difficult to obtain if the system dimension is high or the constraints are complex. The lattice piecewise affine(PWA) function was used to represent the control law of explicit MPC, although the online calculation complexity is reduced, the offline calculation is still prohibitive for complex problems. In this paper, we use the sample points in the feasible region with their corresponding affine functions to construct the lattice PWA approximation of the optimal MPC controller designed for satellite attitude control. The asymptotic stability of satellite attitude control system under lattice PWA approximation has been proven, and simulations are executed to verify that the proposed method can achieve almost the same performance as linear online MPC with much lower online computational complexity and use less fuel than LQR method.



\end{abstract}
\saythanks

\section{Introduction}
    \lettrine{I}{n} recent decades, aerospace technology has developed rapidly while the attitude control technology of spacecraft like satellites is the basis for executing complex and delicate space missions such as observation and on-orbit services. In many applications, the attitude and angular velocity of the spacecraft must be able to track the time-varying trajectories or maintain at preset values, where the main difficulties faced are various interferences from the space environment and highly coupled dynamics of spacecraft.


    PID control method is commonly used for satellite attitude control\cite{show2002spacecraft}, while effective, new methods with distinct advantages are continually being proposed. Chen\cite{chen1993sliding} proposed a robust sliding mode controller to address uncertainty. Su\cite{su2014velocity} proposed an alternative design for velocity-free asymptotic attitude stabilization of rigid spacecraft under the influence of actuator constraints. Benziane\cite{benziane2016velocity} designed a feedback approach for attitude control that only uses a body-referenced vector. Mackunis\cite{mackunis2016adaptive} proposed an attitude controller based on an adaptive neural network which was analyzed through the Lyapunov theorem. Liang\cite{liang2019observer} proposed an observer-based output feedback fault-tolerant control scheme.

    Due to the high cost of spacecraft, long mission period, and requirements of improving efficiency and accuracy as much as possible under external interference, controllers have to operate reliably at the limit of achievable performance\cite{eren2017model}. Model predictive control(MPC) that can provide expected performance and security while satisfying constraints has become a very promising controller. Weiss\cite{weiss2015model} applied MPC for station keeping and momentum management in geosynchronous satellites, achieving favorable results. Caverly\cite{caverly2018split} extended MPC to control position, attitude, and momentum wheels concurrently, demonstrating fuel savings in a one-year simulation experiment.

    MPC or receding horizon control, is a form of control in which the current action is obtained by solving a finite horizon open-loop optimal control problem at each sampling instant. The challenges associated with online computation make the application of MPC difficult for high-frequency systems like satellite attitude control which with limited computing power. Bemporad\cite{bemporad2002explicit} introduced explicit MPC in 2002, showing that the solution function of a multi-parametric quadratic programming(mpQP) problem is piecewise affine(PWA), which allows the online calculation of the optimization problem to be transformed into a table lookup process, thereby reducing the time required for online optimization. Hegrenas\cite{hegrenaes2005spacecraft} used explicit MPC for attitude control of micro-satellites to solve the problem of insufficient computing power.

    While explicit MPC effectively reduces online calculation time, it faces challenges when dealing with more complex optimization problems. The complexity often leads to a significant increase in difficulty for offline calculation of explicit solutions. \cite{wen2009analytical} proposed a solution by utilizing a lattice PWA function to analytically represent the solution of explicit MPC and introduced a method for removing redundant parameters from the lattice PWA function. \cite{xu2016irredundant} established the necessary and sufficient conditions, along with related algorithms, for the irredundant lattice PWA representation. This advancement was subsequently applied to representing the solution of explicit MPC. Nevertheless, as system dimensions and constraints increase, the explicit solution becomes exceedingly complex and challenging to represent. In light of these considerations, this paper intends to employ a lattice PWA function to approximate the explicit solution obtained by KKT conditions. Instead of using a lattice PWA function to represent the entire explicit control law, we only use regions of interest to construct an approximated function. The simulations are conducted to assess the efficacy of the lattice PWA approximation in increasing online computing speed.

    The remainder of this paper is structured as follows: Section \ref{chap:2} describes the dynamic model of satellite attitude and formulate the satellite attitude control as an MPC problem. Section \ref{chap:3} describes the process of getting explicit control laws, the construction of lattice PWA approximation of explicit control law, and analysis of the system stability. Section \ref{chap:4} presents the simulation results of linear MPC, lattice PWA approximation, and LQR of satellite attitude control to show the advantages of our method in terms of online computing speed and fuel savings.

\section{Problem formulation of satellite attitude control}\label{chap:2}

The control structure for attitude control is illustrated in Fig. \ref{fig:control_diagram}. Before formulating the dynamic model of satellite attitude, it is essential to explain the mathematical description of the satellite's attitude. Let $\mathcal{F}_b$ denote the satellite body coordinate system, with its $X$, $Y$, and $Z$ axes aligning with the inertial principal axes of the satellite. Concurrently, $\mathcal{F}_o$ denote the orbital coordinate system of the satellite, $\mathcal{F}_o$'s origin is situated at the satellite's center of mass, $OZ_o$ is directed towards the Earth's center, $OX_o$ is perpendicular to the axis of $OZ_o$ in the orbital plane, pointing forward and the orientation of $OY_o$ is determined by the right-hand rule.

Mathematically, the attitude of the satellite is the attitude of the $\mathcal{F}_b$ coordinate system relative to the $\mathcal{F}_o$ coordinate system, which can also be described as Euler angles. The objective of attitude control in this paper is to maintain the Euler angles of the satellite at predetermined values and saving fuels as much as possible. In Section \ref{chap:2}, we will construct the dynamic model of satellite attitude and formulate the linear MPC problem of satellite attitude control. In the next section, we will generate sample points $p_i$ in the feasible region and calculate the corresponding explicit control law to construct the lattice PWA approximation. It is important to note that the specific details of the actual attitude measurement process are omitted in this paper.
    
\begin{figure}[htbp]
\centering
\includegraphics[width = 0.9\linewidth]{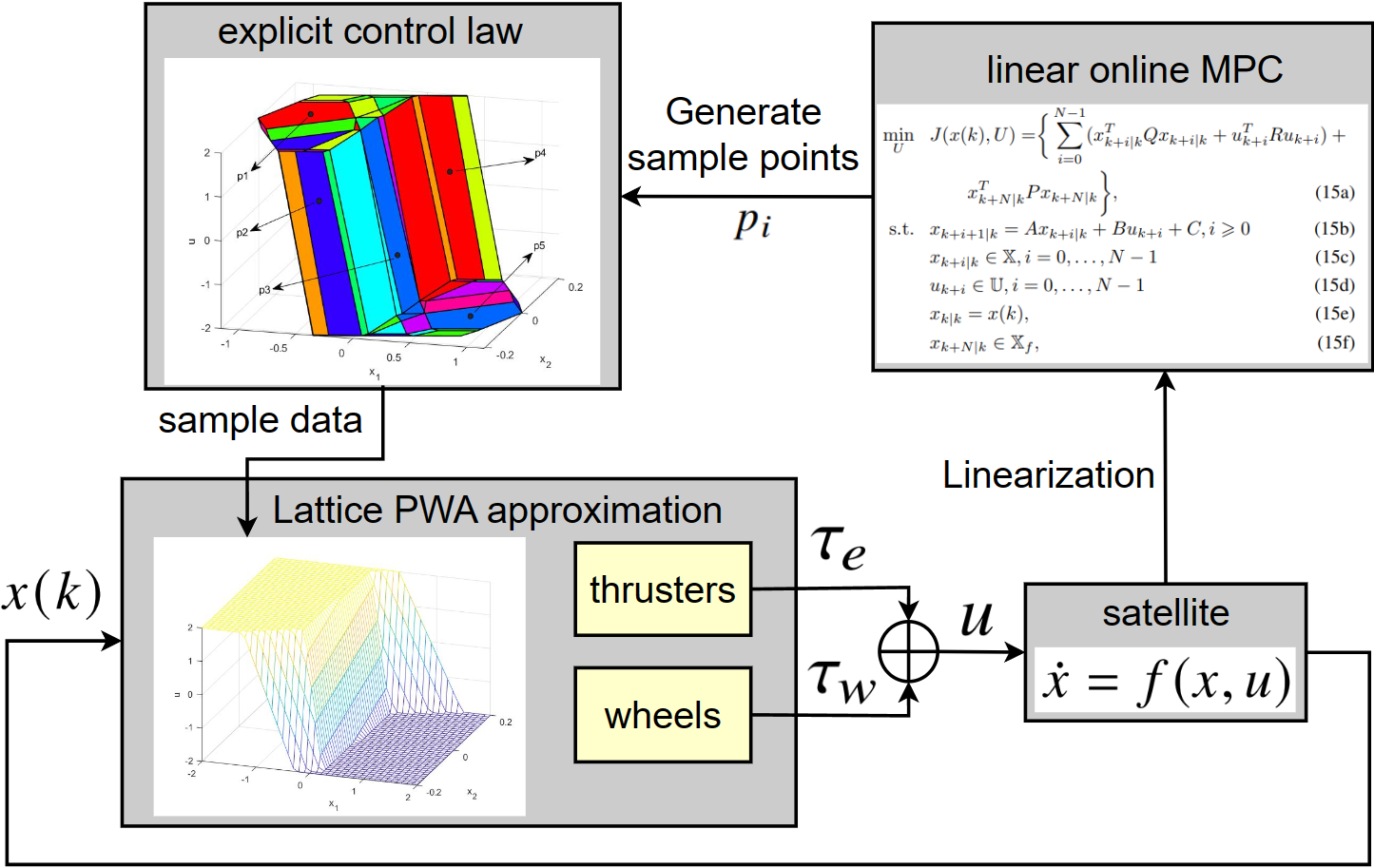}
\caption{Satellite attitude control diagram}
\label{fig:control_diagram}
\end{figure}

\subsection{Dynamics of satellite attitude}

Euler parameters are established through a definition rooted in rotation angle $\theta$ and rotation axis $p$, with a mapping that is defined as
\begin{equation}\notag
\eta=\mathrm{cos}\frac{\theta}{2}, \epsilon=p\mathrm{sin}\frac{\theta}{2},
\end{equation}
the corresponding rotation matrix is
\begin{equation}\label{eq:k1}
R(\eta,\epsilon)=\boldsymbol{1}+2\eta\epsilon^{\times}+2\epsilon^{\times}\epsilon^{\times},
\end{equation}
where $(\cdot)^{\times}$ represents the cross product operator that is defined as
\begin{equation}\notag
\epsilon^{\times}=\begin{bmatrix}
    0&-\epsilon_3&\epsilon_2\\
    \epsilon_3&0&-\epsilon_1\\
    -\epsilon_2&\epsilon_1&0
\end{bmatrix}
\end{equation}

As $R(\eta,\epsilon)\in SO(3)$, we have
\begin{equation}\label{eq:k2}
\dot{R}^b_o=(\omega^b_{bo})^{\times}R^b_o=-(\omega^b_{ob})^{\times}R^b_o,
\end{equation}
where $\omega^b_{bo}$ is defined as the angular velocity of $\mathcal{F}_b$ relative to $\mathcal{F}_o$, expressed in $\mathcal{F}_b$, and $R^b_o$ represents the rotation matrix from $\mathcal{F}_b$ to $\mathcal{F}_o$. By referring to (\ref{eq:k1}), the kinematic differential equation can be derived as
\begin{subequations}
\begin{eqnarray}
& &\dot{\eta}=-\frac{1}{2}\epsilon^T\omega^b_{bo}\label{eq:k15},\\
& &\dot{\epsilon}=\frac{1}{2}[\eta\boldsymbol{1}+\epsilon^{\times}]\omega^b_{bo}\label{eq:k15.1}.
\end{eqnarray}
\end{subequations}

The following content refers to \cite{hegrenaes2005spacecraft} to establish the attitude dynamics model of a micro-satellite. The motion equation of the wheel group can be written as
\begin{subequations}
\begin{eqnarray}
& &\dot{h}_b=\tau_e-[J^{-1}(h_b-\Lambda h_w)]\times h_b,\label{eq:d1}\\
& &\dot{h}_w=\tau_w,
\end{eqnarray}
\end{subequations}
where $h_w$ denotes the $L \times 1$ vector of the axial angular momentum of the wheels, $\tau_e$ is the $3 \times 1$ vector of the external torque, $\tau_w$ is the $L \times 1$ vector of the internal axial torques, and $\Lambda$ is the $3 \times L$ matrix with its columns comprising the axial unit vectors of the $L$ momentum exchange wheels. Let $\omega^b_{ib}$ signify the angular velocity of the satellite coordinate system $\mathcal{F}_b$ relative to the inertial coordinate system $\mathcal{F}_i$. Subsequently, the vector $h_b$ represents the total angular momentum of the spacecraft in the body coordinate system, expressed as
\begin{equation}\label{eq:d2}
h_b=J\omega^b_{ib}+\Lambda h_w,
\end{equation}
where $J$ is an inertia-like matrix, defined as
\begin{equation}\notag
J\triangleq I-\Lambda I_w\Lambda^T,
\end{equation}
the matrix $I$ represents the inertia tensor of the spacecraft, including the flywheel. Additionally, the matrix $I_w = \text{diag}\{I_{w1}, I_{w2}, \ldots, I_{wL}\}$ incorporates the axial momentum of inertia of the flywheel. The axial angular momentum of the flywheel can be expressed in terms of the satellite angular velocity and the axial angular velocity of the flywheel relative to the satellite, denoted as $\omega_w$, this relationship can be formulated as
\begin{equation}\label{eq:d3}
h_w=I_w\Lambda^T\omega^b_{ib}+I_w\omega_w,
\end{equation}
where $\omega_w=[\omega_{w1},\omega_{w1},\dots,\omega_{wL}]^T$ is a $L\times1$ vector, and these relative angular velocities are measured by tachometers fixed on the satellite. 
 
(\ref{eq:d1}) can also be expressed in terms of angular velocity, by defining $\mu\triangleq[h_b,h_w]^T$ and $v\triangleq[\omega^b_{ib},\omega_w]^T$, we can write (\ref{eq:d2}) and (\ref{eq:d3}) into a compact form like
\begin{equation}\label{eq:d4}
\mu=\Gamma v,\hspace{2em}\quad\hspace{2em}\Gamma=
\begin{bmatrix}
I&\Lambda I_w\\
I_w\Lambda^T&I_w
\end{bmatrix}.
\end{equation}

By combining the matrix inversion lemma and (\ref{eq:d4}), we get
\begin{equation}\notag
\begin{bmatrix}
\dot{\omega^b_{ib}}\\
\dot{\omega_w}
\end{bmatrix}=
\begin{bmatrix}
J^{-1}&-J^{-1}\Lambda\\
-\Lambda^TJ^{-1}&\Lambda^TJ^{-1}\Lambda+I_w^{-1}
\end{bmatrix}
\begin{bmatrix}
\dot{h}_b\\
\dot{h}_w
\end{bmatrix},
\end{equation}
where
\begin{subequations}
\begin{eqnarray}\label{eq:d5}
&&\dot{\omega^b_{ib}}=J^{-1}[-(\omega^b_{ib})^{\times}(I\omega^b_{ib}+\Lambda I_w\omega_w)+\tau_e]-\Lambda\tau_w,\\
&&\dot{\omega_w}=\Lambda^TJ^{-1}[(\omega^b_{ib})^{\times}(I\omega^b_{ib}+\Lambda I_w\omega_w)-\tau_e]+[\Lambda^TJ^{-1}\Lambda+I_w^{-1}]\tau_w.
\end{eqnarray}
\end{subequations}

As can be seen from (\ref{eq:d5}), the angular velocity is given in terms of $\mathcal{F}_b$ relative to $\mathcal{F}_i$, while the kinematics are given relative to $\mathcal{F}_o$. However, it would be better if we could describe the attitude of $\mathcal{F}_b$ relative to $\mathcal{F}_o$. This can be achieved through the following equations
\begin{subequations} 
\begin{eqnarray}\label{eq:d6}
&&\omega^b_{ib}=\omega^b_{ob}+R^b_o\omega^o_{io},\\
&&\dot{\omega^b_{ib}}=\dot{\omega^b_{ob}}+\dot{R^b_o}\omega^o_{io},
\end{eqnarray}
\end{subequations}
where $\omega^o_{io}=[0,-\omega_0,0]^T$, $\omega_0$ is assumed to be a constant, equal to the average angular velocity of $\mathcal{F}_o$, and expressed in $\mathcal{F}_i$. This means that the orbit is circular. 
 
Through (\ref{eq:k2}) and (\ref{eq:d6}), we can rewrite (\ref{eq:d5}) as
\begin{subequations}
\begin{eqnarray}
& &\dot{\omega}^b_{ob}=\hat{f}_{\mathrm{inert}}+\hat{f}_{\tau}+\hat{f}_{g}+\hat{f}_{\mathrm{add}},\label{eq:d65}\\
& &\dot{\omega}_{w}=\bar{f}_{\mathrm{inert}}+\bar{f}_{\tau}+\bar{f}_{g},\label{eq:d65.1}
\end{eqnarray}
\end{subequations}
and variables in (\ref{eq:d65}) and (\ref{eq:d65.1}) can be stated as
\begin{subequations}
\begin{eqnarray}\label{eq:d7}
& &\hat{f}_{\mathrm{inert}}=J^{-1}[-(\omega^b_{ob}-\omega_0c_2)^{\times}\times(I[\omega^b_{ob}-\omega_0c_2]+\Lambda I_w\omega_w)],\notag\\
& &\bar{f}_{\mathrm{inert}}=\Lambda^TJ^{-1}[-(\omega^b_{ob}-\omega_0c_2)^{\times}\times(I[\omega^b_{ob}-\omega_0c_2]+\Lambda I_w\omega_w)],\notag\\
& &\hat{f}_{\tau}=J^{-1}\tau-J^{-1}\Lambda\tau_w,\notag\\
& &\bar{f}_{\tau}=-\Lambda^TJ^{-1}\tau+[\Lambda^TJ^{-1}\Lambda+I^{-1}_w]\tau_w,\notag\\
& &\hat{f}_{g}=J^{-1}[3\omega^2_0c_3\times(Ic_3)],\notag\\
& &\bar{f}_{g}=-\Lambda^TJ^{-1}[3\omega^2_0c_3\times(Ic_3)],\notag\\
& &\hat{f}_{\mathrm{add}}=\omega_0\dot{c}_2.\notag
\end{eqnarray}
\end{subequations}
where $c_i$ denotes the $i$th column of the rotation matrix $R^b_{o}$.

Choosing the state vector of the system as $x=[\omega^b_{ob},\omega_{w},\epsilon]^T$, and the input vector is $u\triangleq[\tau^T,\tau_w^T]^T=[\tau_1,\tau_2,\tau_3,\tau_w]^T$, according to (\ref{eq:k15}) and (\ref{eq:d65}), we can have the dynamics model of the system as 
\begin{equation}\label{eq:z1}
	\dot{x}=f(x,u),
\end{equation}
where $\omega^b_{ob}\triangleq[\omega_1,\omega_2,\omega_3]^T$ is the angular velocity of the satellite relative to the orbital coordinate system, $\omega_{w}$ is the angular velocity of the flywheel, $\epsilon\triangleq[\epsilon_1,\epsilon_2,\epsilon_3]^T$ and $\eta$ together constitute the Euler parameters.

\subsection{Linear MPC problem for satellite attitude control}

The objective of attitude control is to drive the satellite's state to expected values while satisfying the state, input, system dynamics constraints, and saving fuels as much as possible. As we adopt a linear MPC approach in this paper, it is necessary to linearize the nonlinear system model. By choosing equilibrium point $x_s = \boldsymbol{0}$, $u_s = \boldsymbol{0}$ as the linearization point, the linear model can be obtained through the first-order Taylor expansion of (\ref{eq:z1}), which can be stated as
\begin{equation}\label{eq:linearplant}
\dot{x}(t)=A_cx(t)+B_cu(t)+C_c,
\end{equation}
where $t$ is a continuous time variable, $A_c=\left.\frac{\partial f(x,u)}{\partial x}\right|_{x=x_s,u=u_s}$, $B_c=\left.\frac{\partial f(x,u)}{\partial u}\right|_{x=x_s,u=u_s}$, $C_c=f(x_s,u_s)-\left.\frac{\partial f(x,u)}{\partial x}\right|_{x=x_s,u=u_s}x_s-\left.\frac{\partial f(x,u)}{\partial u}\right|_{x=x_s,u=u_s}u_s$. Assuming $T_s$ to be the sampling interval, (\ref{eq:linearplant}) is then converted into an equivalent discrete-time form as
\begin{equation}\label{eq:Ldiscreteplant}
x(k+1)=Ax(k)+Bu(k)+C,
\end{equation}	
where $k$ is a discrete time variable, $x(k)\in \mathbb{R}^n$ is the state variable, $u(k)\in\mathbb{R}^m$ is the input variable, $A\in\mathbb{R}^{n\times n}$,$B\in\mathbb{R}^{n\times m}$, $C\in\mathbb{R}^{n\times 1}$.

The linear MPC problem associated with attitude control at sampling time $k$ can be expressed as the following optimization problem:
\begin{subequations}
\begin{eqnarray}
\min_U\hspace{-1.2em}&& J(x(k),U)=\sum^{N-1}_{i=0}(x_{k+i|k}^TQx_{k+i|k}+u_{k+i}^TRu_{k+i})+x_{k+N|k}^TPx_{k+N|k},\label{mpc1:0}\\
\mathrm{s.t.}\hspace{-1.2em}
&&x_{k+i+1|k}=Ax_{k+i|k}+Bu_{k+i}+C, i \geqslant 0\label{mpc1:1}\\
&&x_{k+i|k}\in\mathbb{X}, i = 0,\dots,N\label{mpc1:3}\\
&&u_{k+i}\in\mathbb{U}, i = 0,\dots,N-1\label{mpc1:4}\\
&&x_{k|k}=x(k),\label{mpc1:2}\\
&&x_{k+N|k}\in\mathbb{X}_f,\label{mpc1:5}
\end{eqnarray}
\end{subequations}
where $U=[u_k^T,\dots,u_{k+N-1}^T]^T$ is the optimized variable, $x(k)$ is the system state measured at time instant $k$, $x_{k+i|k}\in\mathbb{R}^n$, $u_{k+i} \in \mathbb{R}^m$ denote the predicted state and input variables at time step $k+i$, $Q=Q^T\succcurlyeq0$, $R =R^T\succ0$ and $P\succ0$ are the weighting matrices of the state, input and terminal state variables, $N$ is the prediction horizon, (\ref{mpc1:3}) and (\ref{mpc1:4}) denote the linear constraints on the state and input respectively, which that can be expressed as $x_{\min} \leqslant x_{k+i|k}\leqslant x_{\max}$ and $u_{\min} \leqslant u_{k+i}\leqslant u_{\max}$, $\mathbb{X}_f\subset\mathbb{X}$ is a terminal constraint. After the optimization problem (\ref{mpc1:0}) is solved, the first item $u_k^*$ of the optimal sequence $U^*=[(u_k^*)^T,\dots,(u_{k+N-1}^*)^T]^T$ is applied to the plant, this implicit MPC law can be stated as $k_0(x(k)) = u_k^*$, and the corresponding value function is defined as $J^*(x(k)) = J(x(k),U^*)$. 

Function (\ref{mpc1:0}) indicates that the optimization problem of attitude control is a quadratic programming(QP) problem. As mentioned before, lattice PWA approximation is constructed by the control law of explicit MPC, and the process of getting explicit control law need us to convert the QP problem (\ref{mpc1:0}) into a mpQP form as 
\begin{subequations}
\begin{eqnarray}\label{youhua2}
\min_z\hspace{-1em}&& \frac{1}{2}z^THz,\\
\mathrm{s.t.}\hspace{-1em}&&Gz\leqslant W+Sx(k),
\end{eqnarray}
\end{subequations}
where $H$, $G$, $W$, $S$ are matrices  calculated from (\ref{mpc1:0})-(\ref{mpc1:5}).

\section{Lattice PWA approximation of explicit MPC}\label{chap:3}
\subsection{Explicit MPC}
Explicit MPC calculates explicit solutions of mpQP, transforms the online calculation into a table lookup process. One important feature of the explicit control law is PWA, which is illustrated in Lemma \ref{lem:expmpc}.
\begin{lemma}[\cite{bemporad2002explicit}]\label{lem:expmpc}
    Consider the mpQP problem (\ref{youhua2}) and l et $H\succ0$. Then the set of feasible parameters $\Omega\subset\mathbb{X}$ is convex, the optimizer $U^*(x):\Omega\mapsto\mathbb{R}^{N_c\cdot m}$ is continuous PWA, and the optimal solution $J^*(x):\Omega\mapsto\mathbb{R}$ is continuous, convex and piecewise quadratic.
\end{lemma}

Then we define a continuous PWA function as shown in Definition \ref{def:PWA}:

\begin{definition}\label{def:PWA}
A function $f:D\mapsto\mathbb{R}$, where $D\subset\mathbb{R}^n$ is convex, is continuous PWA when it satisfies the following conditions:

(1) $D$ can be divided into a finite number of nonempty convex regions, i.e., $D=\bigcup^{I_O}_{i=1}O_i$, $O_i\cap O_{i'}=\emptyset(i \ne i')$, $O_i \ne \emptyset$, $I_O$ is the number of regions.

(2) $f$ is defined by hyperplane $l_{\rm loc(i)}:O_i\mapsto\mathbb{R}$ in each region $O_i$, i.e.
\begin{equation}\label{eq:loc}
    f(x)=l_{\rm loc(i)}(x)=\sum_{j=1}^{n}a_{ij}x_j+b_i, \forall x\in O_i
\end{equation}

(3) $f$ is continuous at the boundary of any two regions, i.e., 
\begin{equation}\notag
l_{\rm loc(i)}(x)=l_{\rm loc(i')}(x),\forall x \in O_i\cap O_{i'}.
\end{equation}
\end{definition}

Explicit MPC divides the control process into offline and online calculations. In the offline process, the state space of the system is divided into critical regions by solving the mpQP problem (\ref{youhua2}), and the explicit control rules corresponding to each region can be obtained through KKT conditions. The online calculation is changed from solving a complex QP problem to a simple table lookup process: at each moment, the current state is first obtained to determine the critical region where the system is in, then calculations are performed based on the explicit expression corresponding to the region.
 
First, we illustrate the process of offline calculation. The KKT conditions of mpQP (\ref{youhua2}) can be listed as
\begin{subequations}
\begin{eqnarray}\label{kkt}
& &Hz^*+G^T_{\mathcal{A}^*}\lambda^*+G^T_{\mathcal{N}^*}\mu ^*=0,\\
& &G_{\mathcal{A}^*}z^*=W_{\mathcal{A}^*}+S_{\mathcal{A}^*}x,\label{act}\\
& &G_{\mathcal{N}^*}z^*<W_{\mathcal{N}^*}+S_{\mathcal{N}^*}x,\label{inact}\\
& &\lambda^*\geqslant0,\label{d}\\
& &\mu^*\geqslant0,\\
& &\lambda^{*T}(G_{\mathcal{A}^*}z^*-W_{\mathcal{A}^*}-S_{\mathcal{A}^*}x)=0,\\
& &\mu^{*T}(G_{\mathcal{N}^*}z^*-W_{\mathcal{N}^*}-S_{\mathcal{N}^*}x)=0,
\end{eqnarray}
\end{subequations}
where $x$ represents $x(k)$ for simplicity, (\ref{act}) and (\ref{inact}) respectively represent the active constraints and inactive constraints of the optimal solution $z^*$, and $G_j, W_j, S_j$ respectively represent the $j$-th row of $G$, $W$, $S$. The indices of the effective constraints and the ineffective constraints can be expressed as a set $\mathcal{A}^*$ and $\mathcal{N}^*$ respectively. For a fixed effective constraint set ${\mathcal{A}^*}$, if $G_{\mathcal{A}^*}$ is of full rank, then we have
\begin{subequations}
\begin{eqnarray}
& &z^*=H^{-1}G^T_{\mathcal{A}^*}(G_{\mathcal{A}^*}H^{-1}G^T_{\mathcal{A}^*})^{-1}(W_{\mathcal{A}^*}+S_{\mathcal{A}^*}x),\label{econtrollaw1}\\
& &-(G_{\mathcal{A}^*}H^{-1}G^T_{\mathcal{A}^*})^{-1}(W_{\mathcal{A}^*}+S_{\mathcal{A}^*}x)\geqslant0,\label{econtrollaw2}\\
& &GH^{-1}G^T_{\mathcal{A}^*}(G_{\mathcal{A}^*}H^{-1}G^T_{\mathcal{A}^*})^{-1}(W_{\mathcal{A}^*}+S_{\mathcal{A}^*}x)\leqslant W+Sx,\label{econtrollaw3}
\end{eqnarray}
\end{subequations}
where (\ref{econtrollaw1}) indicates that the optimal solution $z^*$ is an affine expression of $x$, and the critical region $CR_i$ corresponding to the active constraint set can be obtained through (\ref{econtrollaw2}) and (\ref{econtrollaw3}). Suppose there are $N_{\mathrm{mpc}}$ critical regions, after linear transformation from $z^*$ to $U^*$, the explicit control law $U_i(x)$ of region $CR_i$ can be written as 
\begin{subequations}
\begin{eqnarray}
& &U_i=f_ix+g_i,x\in CR_i,\label{eq:explaw1}\\
& &CR_i:H_ix\leqslant K_i,i=1,\dots,N_{\mathrm{mpc}},\label{eq:explaw2}
\end{eqnarray}
\end{subequations}
It is noted that all the $N_{\mathrm{mpc}}$ critical regions have to be traversed to determine (\ref{eq:explaw1}) - (\ref{eq:explaw2}), and the obtained $U_i$ and $CR_i$ are stored in a lookup table. During online calculation, for a current state $x$, corresponding $f_i$ and $g_i$ can be searched in the lookup table, and the optimal input can be written as
\begin{equation}
    u^*(x) = [\mathbf{1}_m,\mathbf{0}_{m\times N_u -m}]\cdot(f_ix+g_i),
\end{equation}
where $\mathbf{1}_m\in\mathbb{R}^m$ is all-one vector, $\mathbf{0}_{m\times N_u -1}\in\mathbb{R}^{m\times N_u -1}$ is all-zero vector. It is important to note that the explicit solution $u^*(x)$ obtained through the KKT conditions is PWA, which is described in Lemma \ref{lem:expmpc}.

\subsection{Lattice PWA function}\label{chap:3.2}
The explicit control law of attitude control can be too complicated to solve, so this paper choose the lattice PWA function to approximate the explicit control law. As early as the 1960s, Wilkinson R H. pointed out that the set of linear functions defined on convex sets and the binary operations $\max$, $\min$ constitute an abstract algebraic lattice that can represent any PWA function\cite{wilkinson1963method}, and the definition of lattice PWA function is as follows.
\begin{lemma}[\cite{wilkinson1963method}]\label{lem:lattice}
Let $f:D\mapsto\mathbb{R}$ be a continuous PWA function , then $f$ can be represented as
\begin{equation}\label{eq:lattice}
f(x)=\max_{i=1,\cdots,N_t}\left\{ \min_{j\in I_{\geqslant,i}}\{l_j(x)\}\right\},\forall x\in \mathbb{R},
\end{equation}
in which $l_j$ is an affine function representing a literal of the lattice PWA representation (\ref{eq:lattice}), $\min_{j\in I_{\geqslant,i}}\{l_j\}$ is called a term , and the index set $I_{\geqslant,i}$ is defined as
\begin{equation}
    I_{\geqslant,i}=\{j|l_j(x)\geqslant l_{\rm loc(i)}(x),x\in O_i\}
\end{equation}
$N_t$ is the number of terms, and $O_i$ is a base region, $l_{\rm loc(i)}(x)$ is the local function of $O_i$ as defined in (\ref{eq:loc}). The base region $O_i$ is a region such that no affine functions intersect with $l_{loc(i)}(x)$ in the interior of $O_i$, i.e., 
\begin{equation}\label{eq:base}
    \{x|l_j(x)=l_{\rm loc(i)}(x), \forall j \ne {\rm loc(i)}\} \cap {\rm int}(O_i) = \emptyset
\end{equation}
 in which ${\rm int}(O_i)$ denotes the interior of $O_i$, the base regions satisfy $O_i\cap O_{i'}=\emptyset(i \ne i')$, $\bigcup^{I_{O'}}_{i=1}O_i=D$, $I_{O'}$ is the number of base regions. It has been shown in \cite{xu2021error} that the base region is a subset of the convex region.
\end{lemma}

To further explain the construction of the lattice PWA representation, we take a one-dimensional PWA function as shown in Fig. \ref{fig:1dimPWA} to illustrate Lemma \ref{lem:lattice}, the domain of $f(x)$ is divided into $O_1,\dots,O_6$ that satisfy (\ref{eq:base}).
\begin{figure}[htbp]
\centering
\includegraphics[width = 0.5\textwidth]{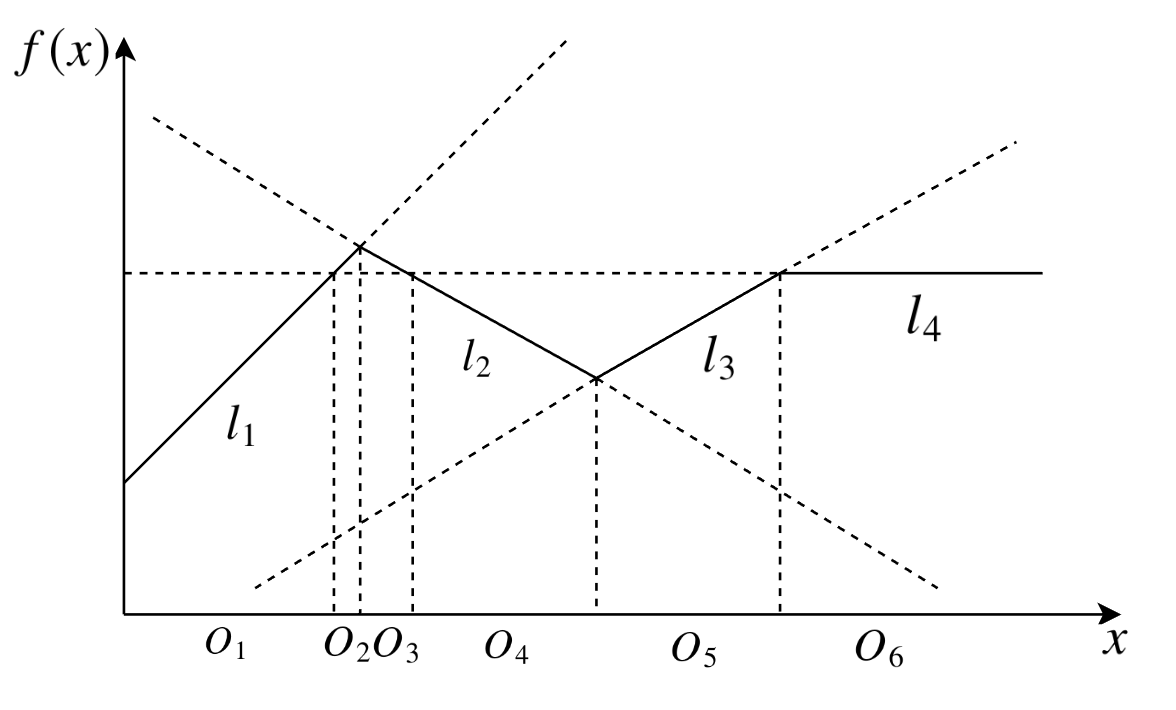}
\caption{One-dimensional PWA function}
\label{fig:1dimPWA}
\end{figure}

Take $O_1$ for example, as $l_2 \geqslant l_1$, $l_4 \geqslant l_1$, $\forall x\in O_1$, we have $I_{\geqslant,1}=\{1,2,4\}$, then the literals in all terms of lattice PWA function corresponding to the PWA function $f(x)$ in Fig. \ref{fig:1dimPWA} can the obtained in the same way, which is shown in Table \ref{tab:1dimlatticeterm}. 
\begin{table}[hbt!]
\caption{\label{tab:1dimlatticeterm} One-dimensional lattice PWA}
\centering\small
\begin{tabular}{ccc}
\hline
Base region & $l_{loc(i)}$  & Term \\\hline
$O_1$ & $l_1$ &  $\mathrm{Term}_1=\min\{l_1,l_2,l_4\}$\\
$O_2$ & $l_1$ &  $\mathrm{Term}_2=\min\{l_1,l_2\}$\\
$O_3$ & $l_2$ &  $\mathrm{Term}_3=\min\{l_1,l_2\}$\\
$O_4$ & $l_2$ &  $\mathrm{Term}_4=\min\{l_1,l_2,l_4\}$\\
$O_5$ & $l_3$ &  $\mathrm{Term}_5=\min\{l_1,l_3,l_4\}$\\
$O_6$ & $l_4$ &  $\mathrm{Term}_6=\min\{l_1,l_3,l_4\}$\\
\hline
\end{tabular}
\end{table}

According to (\ref{eq:lattice}), the lattice PWA function of $f(x)$ can be stated as
\begin{equation}\label{lattice}
f_{\rm lattice}(x)=\max\{\mathrm{Term}_1,\dots,\mathrm{Term}_6\},
\end{equation}
after removing redundant literals and terms \cite{xu2016irredundant}, the irredundant lattice PWA representation is 
\begin{equation}\label{eq:irrelattice}
    f_{\rm lattice}(x)=\max\{\min\{l_1,l_2\},\min\{l_1,l_3,l_4\}\}.
\end{equation}
readers can verify the equivalence of (\ref{eq:irrelattice}) and function depicted in Fig. \ref{fig:1dimPWA}.

\subsection{Lattice PWA approximation of explicit control law}

According to Section \ref{chap:3.2}, we can use a lattice PWA function to represent any continuous PWA function. As the analytical solution of explicit MPC is continuous PWA, a lattice PWA function can be used to represent the solution of explicit MPC. However, as the system dimension increases, the optimization problem becomes more complex, and the number of regions in the explicit solution will increase exponentially, making it extremely difficult to establish a lattice PWA representation \cite{xu2016irredundant}, therefore, this section uses a lattice PWA function to approximate the explicit control law. Unlike representations, approximations do not need information about every region, but focus on regions of interest, and are constructed through sample points in these regions. 

The following simple example illustrates the process of constructing a lattice PWA approximation.

$Example\ \mathit{1}$: Suppose we have a two-dimensional linear system as:
\begin{equation}\label{eq:exsys}
x(k+1)=\begin{bmatrix}
0.7&-0.1\\
0.2&1
\end{bmatrix}x(k)
+\begin{bmatrix}
0.1\\
0.01
\end{bmatrix}u(k)
\end{equation}
with constraints $x_{\min}=[-2,-2]^T$, $x_{\max}=[2,2]^T$, $u_{\min}=-2$, $u_{\max}=2$, the state weight is $Q=[2,0; 0,2]$, the input weight is $R=0.01$.

Using the MPT Toolbox \cite{Herceg_Kvasnica_Jones_Morari_2013}, we could get the optimal control law and visualize it in Fig. \ref{fig:mpc3law},
\begin{figure}
\centering
\subfigure[Explicit control law]{\includegraphics[width = 0.4\textwidth]{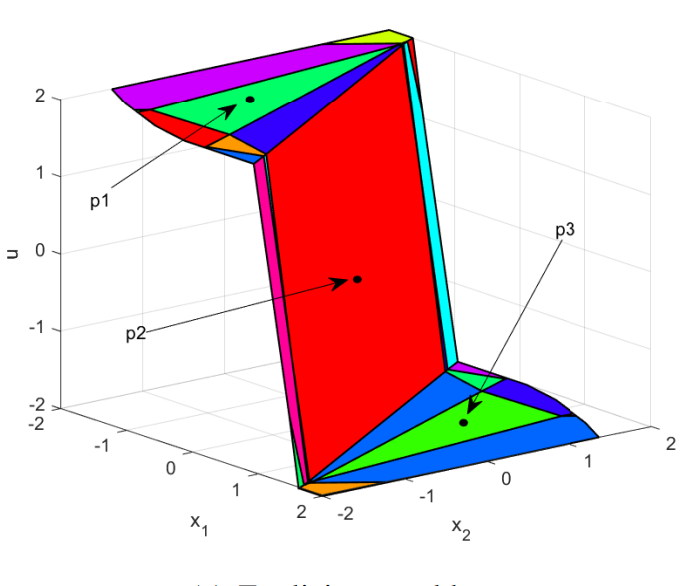}\label{fig:mpc3law}}\hspace{1cm}
\subfigure[Lattice PWA approximation]{\includegraphics[width = 0.4\textwidth]{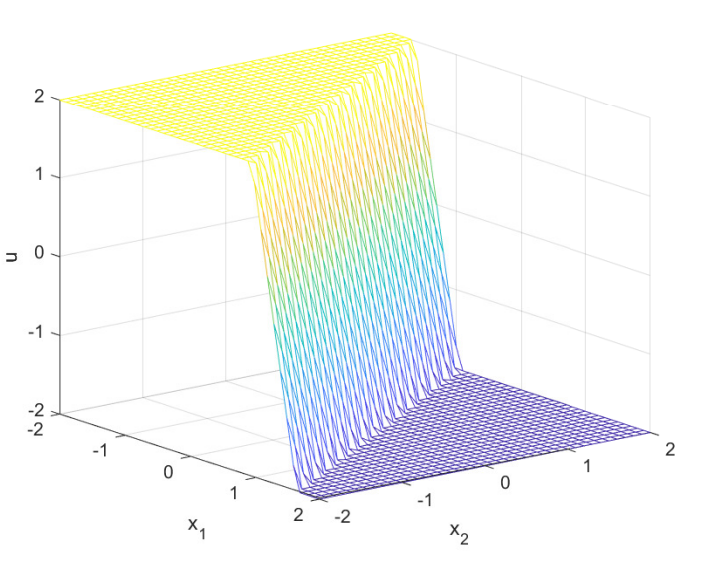}\label{fig:mpc3lattice}}
\caption{Control law}
\label{fig:mpc3controllaw}
\end{figure}
where each color plane represents the affine function corresponding to a critical region in the explicit control law. For system (\ref{eq:exsys}), the sample points $p_1=[-1.5,0.5]^T$, $p_2=[0,0]^T$, $p_3=[1,0.5]^T$ are selected in the state feasible region, for each point, we get the corresponding affine functions through the KKT conditions, which are listed below:
\[l_1=[0,0]^T\cdot x+2 ,\ l_2= [-5.72,-3.73]^T\cdot x,\ l_3=[0,0]^T\cdot x-2\]
and for each point $p_i, i=1, 2, 3$, the order of the affine functions $l_1, l_2, l_3$ is calculated,
\begin{eqnarray}\notag
    & &p_1:l_3\leqslant l_1\leqslant l_2 ,\\ 
    & &p_2:l_3\leqslant l_2\leqslant l_1,\\ 
    & &p_3:l_2\leqslant l_3\leqslant l_1,
\end{eqnarray}
the term corresponding to each sample point is then obtained according to Lemma \ref{lem:lattice},
\begin{eqnarray}\notag
    & &\mathrm{Term}_1=\min\{l_1,l_2\},\\ 
    & &\mathrm{Term}_2=\min\{l_1,l_2\},\\ 
    & &\mathrm{Term}_3=\min\{l_1,l_3\},
\end{eqnarray}
the lattice PWA approximation of the function depicted in Fig. \ref{fig:mpc3law} can be expressed as 
\begin{equation}
    \hat{f}_{\rm lattice} = \max\{\min\{l_1,l_2\},\min\{l_1,l_2\}\min\{l_1,l_3\} \}
\end{equation}
Fig. \ref{fig:mpc3lattice} shows the plot of the lattice PWA approximation, which is very close to the function in Fig. \ref{fig:mpc3law}.

By choosing the initial state as $x(0) = [1,0]^T$, the prediction horizon $N=5$, the simulation results of the system (\ref{eq:exsys}) is shown in Fig. \ref{fig:mpc3simx} and Fig. \ref{fig:mpc3simu}, where lattice PWA approximation achieved almost the same performance as linear online MPC.
\begin{figure}
\centering
\subfigure[State curve]{\includegraphics[width = 0.4\textwidth]{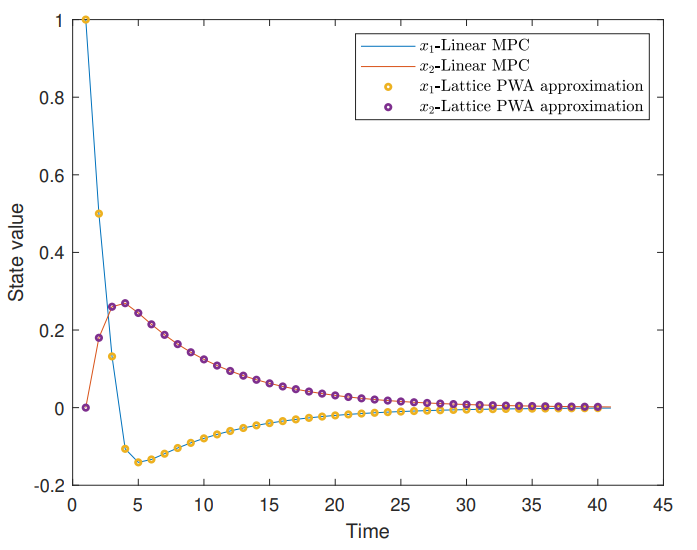}\label{fig:mpc3simx}}\hspace{1cm}
\subfigure[Input curve]{\includegraphics[width = 0.4\textwidth]{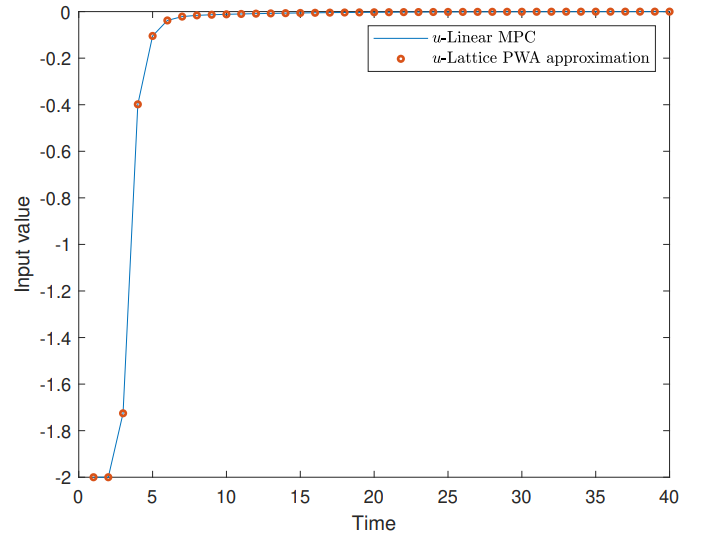}\label{fig:mpc3simu}}
\caption{Simulation curves}
\label{fig:mpc3sim}
\end{figure}

Therefore, the idea of the lattice PWA approximation is to collect the control expressions of those regions that dominate and ignore those unimportant regions, then link these hyperplanes up to obtain the approximated lattice PWA function. And sample points are randomly generated to obtain the corresponding control law until the expected PWA representation can be formulated.

Assumed that the sample data $\mathcal{X}\times\mathcal{U}$ is already collected, $\mathcal{X}$ contains the sample state points $x_i$, $\mathcal{U}$ contains the corresponding control function $u_i(x)$ that is obtained from the KKT conditions of (\ref{youhua2}) and the number of sample points is $N_s$. For each sample point, define the index set $J_{\geq, i}$ as
\begin{equation}
J_{\geq,i}=\{j|u_j(x_i) \geq u_i(x_i)\}.
\end{equation}
It is noted that the set $J_{\geq,i}$ is different from $I_{\geq,i}$ in that $I_{\geq,i}$ holds for all the points in the base region $O_i$. The lattice PWA approximation can be constructed as
\begin{equation}\label{eq:applattice}
  \hat{f}_{\rm lattice}=\max\limits_{i=1,\dots,N_s} \{\min_{j \in J_{\geq,i}}\{u_j\}\}. 
\end{equation}
and the process of constructing the lattice PWA approximation is shown in Algorithm \ref{alg:connect}.

\begin{algorithm}
\caption{Constructing lattice PWA approximation from sample data}
\label{alg:connect}
\begin{algorithmic}[1] 
\REQUIRE Sample data set $\mathcal{X}\times\mathcal{U}$.
\ENSURE Lattice PWA approximation $\hat{f}_{lattice}$.
\FOR{$k=1$ to $m$}
    \FOR{$i=1$ to $N_s$}
        \STATE $x_i\in\mathcal{X}$, $u_i\in\mathcal{U}$
        \STATE $l_{loc(i)}=u_i$;
        \FOR{$j=1$ to $N_s$}
            \STATE $x_j\in\mathcal{X}$, $u_j\in\mathcal{U}$
            \STATE $l_j=u_j$;
            \IF{$l_j(x_i) \geqslant l_{loc(i)}(x_i)$}
                \STATE $\mathrm{Term}_i = [\mathrm{Term}_i, l_j]$;
            \ENDIF
        \ENDFOR
    \ENDFOR
    \STATE $\hat{f}_{\rm lattice}^k = \max\{\min\{\mathrm{Term}_1\},\dots,\min\{\mathrm{Term}_{Ns}\}\}$;
\ENDFOR
\STATE $\hat{f}_{\rm lattice} = [\hat{f}_{\rm lattice}^1,\dots,\hat{f}_{\rm lattice}^m]$;
\end{algorithmic}
\end{algorithm}
Given the following assumption, the equivalence of lattice PWA approximation and the explicit control law is illustrated in \cite{xu2021error}.
\begin{assumption}\label{ass:latticeterm}
Assuming that all the distinct affine functions have been sampled in the domain of interest $\mathbb{X}$.
\end{assumption}
\begin{lemma}[\cite{xu2021error}]\label{lem:errlattice}
Supposing that 
\begin{equation}\notag
    \epsilon_d = \max\left\{ \max\{\min_{j\in J_{\geqslant,i}}\{ u_j(x)\}\} - \min\{\max_{j\in J_{\leqslant,i}}\{ u_j(x)\}\} \right\}
\end{equation}
where
\begin{equation}
J_{\leq,i}=\{j|u_j(x_i) \leq u_i(x_i)\}.
\end{equation}
If Assumption \ref{ass:latticeterm} holds, then we have 
\begin{equation}\notag
    |\max\{\min_{j\in J_{\geqslant,i}}\{ u_j(x)\}\}-u^*(x)|\leqslant \epsilon_d
\end{equation}
Furthermore, if $\epsilon_d=0$, we have
\begin{equation}
\hat{f}_{\rm lattice}(x)=u^*(x),\forall x\in \mathbb{X}
\end{equation}
\end{lemma}
According to Assumption \ref{ass:latticeterm} and Lemma \ref{lem:errlattice}, the error between lattice PWA approximation and the explicit control law is bounded.

\subsection{Stability analysis}
For the convenience of subsequent proof, we will adopt the following form of optimization problem
\begin{subequations}
\begin{eqnarray}
\min_U\hspace{-1.2em}&& J(x(k),U)=\sum^{N-1}_{i=0}l(x_{k+i|k},u_{k+i})+F(x_{k+N|k}),\label{mpc4:0}\\
\mathrm{s.t.}\hspace{-1.2em}
&&x_{k+i+1|k}=\hat{f}(x_{k+i|k},u_{k+i}),i \geqslant 0\label{mpc4:1}\\
&&(\ref{mpc1:2})-(\ref{mpc1:5}),\notag
\end{eqnarray}
\end{subequations}
where $l(x_{k+i|k},u_{k+i})$ and $F(x_{k+N|k})$ denote the quadratic function in (\ref{mpc1:0}), (\ref{mpc4:1}) is the linear discrete dynamics function (\ref{eq:Ldiscreteplant}), and denote $k_f(\cdot)$ as a local controller, the set of states that can be controlled by MPC with fixed horizon $N$ is $\mathbb{X}_N$, which satisfies
\begin{equation}\notag
    \mathbb{X}_N=\{x(0)\in\mathbb{X}|\exists u(k)\in\mathbb{U}, k= 1,\dots,N-1,\mathrm{such\ that}\ x(k)\in\mathbb{X},k=1,\dots,N-1,x(N) \in \mathbb{X}_f\}
\end{equation}

The stability of the original nonlinear system under Lattice PWA approximation is proven through the stability of the nominal system (\ref{mpc4:1}). Before the analysis, we define the $K_{\infty}$ functions.

\begin{definition}[\cite{ellis2014tutorial}]\label{def:kfunction}
A continuous function $\alpha : [0,\alpha )\mapsto[0,\infty )$ belongs to class $K$ if it is strictly increasing and satisfies $\alpha(0)=0$. A class $K$ function $\alpha$ is called a class $K_\infty$ function if $\alpha$ is unbounded.
\end{definition}

The nominal model of the system considered in this study can be expressed as
\begin{equation}\label{eq:nominalmodel}
x(k+1)=\hat{f}(x(k),u(k),0).
\end{equation}
where $\hat{f}$ is derived in (\ref{mpc4:1}), $x(k)$ is the system state vector at sample time $k$, and $u(k)$ is the input vector. The real system can be expressed as 
\begin{equation}\label{eq:realmodel}
z(k+1)=\hat{f}(x(k),u(k),\omega(k)).
\end{equation}
where $z(k)$ is the real state, $\omega(k)$ is the system disturbance\cite{huang2023economic}, we assume that $\omega(k) \in \mathbb{W}$.

Given Assumption \ref{ass:xf}, the asymptotical stablility of the nominal system (\ref{eq:nominalmodel}) under the MPC framework (\ref{mpc4:0}) is given.

\begin{assumption}\label{ass:xf}
$\mathbb{X}_f\subset\mathbb{X}$, $\mathbb{X}_f$ is closed, $\mathbf{0}\in\mathbb{X}_f$, $k_f(x)\in\mathbb{U}$ for all $x\in\mathbb{X}_f$, $\hat{f}(x,k_f(x))\in\mathbb{X}_f$ for all $x\in\mathbb{X}_f$, $J^*(\hat{f}(x,k_f(x))) - J^*(x)+l(x,k_f(x))\leqslant 0$ for all $x\in\mathbb{X}_f$.
\end{assumption}

\begin{lemma}[\cite{mayne2000survey}]\label{lem:mpc4sta}
If Assumption \ref{ass:xf} holds, then system (\ref{mpc4:1}) is asymptotically stable under MPC framework (\ref{mpc4:0}) as long as $x(k)\in\mathbb{X}_N$, and $J^*(x(k))$ can be regarded as a Lyapunov function of the closed-loop system.
\end{lemma}

\begin{lemma}\label{lem:errlattice2}
If Assumption \ref{ass:latticeterm}-\ref{ass:xf} hold, then the lattice PWA approximation $\hat{f}_{\rm lattice}(x)$ asymptotically stabilizes the system (\ref{eq:nominalmodel}).
\end{lemma}

\begin{proof}
According to Lemma \ref{lem:errlattice2}, if Assumption \ref{ass:latticeterm} holds, then the lattice PWA approximation is equivalent to the MPC control law when $\epsilon_d = 0$, i.e. 
\begin{equation}
 \hat{f}_{\rm lattice}(x)= u^*(x),\forall x\in \mathbb{X} 
\end{equation}
According to Lemma \ref{lem:mpc4sta}, if Assumption \ref{ass:xf} holds, then $\hat{f}_{\rm lattice}(x) = u^*(x)$ asymptotically stabilizes the system (\ref{eq:nominalmodel}).
\end{proof}

Therefore, it can be concluded that lattice PWA approximation $\hat{f}_{\rm lattice}(x)$ is equivalent to the optimal solution of linear online MPC problem and ensures the stability of the nominal system (\ref{eq:nominalmodel}). Under Assumption \ref{ass:flip}, the asymptotical stability of the real system (\ref{eq:realmodel}) is given.

\begin{assumption}\label{ass:flip}
Let $\varphi = [x^T,u^T]^T$, and the system function $f(\varphi)$ is assumed to be Lipschitz continuous, i.e., there is a constant $L_1$, such that $||f(\varphi_i)-f(\varphi_j)||\leqslant L_1||\varphi_i-\varphi_j||$, for all $\varphi_i$, $\varphi_j$ in the state-input domain.
\end{assumption}

Under Assumption \ref{ass:flip}, the error bound between the nonlinear and linearized function can be derived, as indicated in Lemma \ref{lem:errorbound}.

\begin{lemma}\label{lem:errorbound}
Given that Assumptions \ref{ass:flip} holds, and denote linear discrete function as $\hat{f}(\varphi)=[A, B]\varphi + C$. Suppose that for any $\varphi \in D$, where $D$ is the domain of $f$, there exists a point $\varphi_s$ such that $||\varphi -\varphi_s||\leqslant \sigma$. In this case, the deviation between the nonlinear function $f(\varphi)$ and linearized function $\hat{f}(\varphi)$ satisfies the following inequality:
\begin{equation}\label{eq:errorb}
||f(\varphi)-\hat{f}(\varphi)||\leqslant L\sigma,
\end{equation}
where $L$ is a constant determined by $f$ and $\hat{f}$.
\end{lemma}
\begin{proof}
Obviously, affine function $\hat{f}(\varphi)$ is Lipschitz, according to Assumption \ref{ass:flip}, we have
\begin{equation}\notag
||\hat{f}(\varphi_i)-\hat{f}(\varphi_j)||\leqslant L_2||\varphi_i-\varphi_j||
\end{equation}
where $\varphi_i$, $\varphi_j$ are arbitrary points in the state-input domain, then
\begin{align}
&||f(\varphi)-\hat{f}(\varphi)||\notag\\
&= ||f(\varphi)-f(\varphi_s)+\hat{f}(\varphi_s)-\hat{f}(\varphi)||\notag\\
&\leqslant L_1||\varphi - \varphi_s|| + L_2 ||\varphi - \varphi_s||\notag\\
&= (L_1+L_2)||\varphi - \varphi_s||.
\end{align}
where $\varphi_s = [x_s^T, u_s^T]^T$ represents the linearization point. As $||\varphi-\varphi_s||\leqslant\sigma$, letting $L = L_1 + L_2$, and we have (\ref{eq:errorb}) proved.
\end{proof}

To facilitate the stability analysis of the original system, we define the level set of the Lyapunov function $J^*(\cdot)$ 
\begin{equation}
\Omega_{\rho}=\{x\in\mathbb{X}:J^*(x)\leqslant\rho\}.
\end{equation}

The following theorem demonstrates the stability of the system (\ref{eq:realmodel}).
\begin{theorem}
Consider the system (\ref{eq:realmodel}), if Assumption \ref{ass:xf}-\ref{ass:flip} hold, the magnitude of the partial derivative $\frac{\partial J^*(x)}{\partial x}$ is upper bounded such that $||\frac{\partial J^*(x)}{\partial x}||\leqslant K_V$ for all $x \in \mathbb{X}$. Suppose there exist $\epsilon_s > 0 $ and $\rho_s > 0$ such that 
\begin{equation}\label{realstabl}
-\alpha(\rho_s)+K_VL\delta+ML^2\delta^2\leqslant-\epsilon_s
\end{equation}
where $\alpha(\cdot)$ is a class $K_{\infty}$ function associated with Definition \ref{def:kfunction}, and $M$ is the constant associated with the Taylor expansion of $J^*(x)$. Let
\begin{equation}   \notag
\Omega_{\rho_{\min}}\subset\mathbb{X}_e\subset\Omega_{\rho_{\max}}\subset\mathbb{X}_N\subset\mathbb{X},
\end{equation}
where $\mathbb{X}_e$ is the robust control invariant set \cite{blanchini1999set}, i.e., $\forall x\in\mathbb{X}_e$, $\omega\in\mathbb{W}$, $\exists u\in \mathbb{U}$, such that $\hat{f}(x,u,\omega) \in\mathbb{X}_e$, and $\Omega_{\rho_{\min}}$ is the level set defined as
\begin{equation}\notag
\Omega_{\rho_{\min}} := \{x\in\mathbb{X}|J^*(x)\leqslant\max\{J^*(x(k+1)):||x(k) -x_s||_2\leqslant\rho_s\}\}
\end{equation}
and $\Omega_{\rho_{\max}}$ is the maximum level set within $\mathbb{X}_N$.

In this case, the closed-loop system (\ref{eq:realmodel}) converges to the robust invariant set $\mathbb{X}_e$ in finite steps and is maintained in $\mathbb{X}_e$ for any initial condition $x(0)\in\Omega_{\rho_{\max}}$.
\end{theorem}
\begin{proof}
According to Lemma \ref{lem:errorbound}, if $x(k) = z(k)$, then
\begin{equation}\notag
||z(k+1)-x(k+1)||\leqslant L\sigma,
\end{equation}
Taking the Taylor expansion around $x(k+1)$ yields
\begin{equation}\notag
    J^*(z(k+1)) = J^*(x(k+1))+\frac{\partial J^*(x)}{\partial x}|_{x(k+1)}\cdot(z(k+1)-x(k+1))+O(||z(k+1)-x(k+1)||^2)
\end{equation}
According to Lemma \ref{lem:expmpc}, $J^*(x)$ is piecewise quadratic. For $x\in\mathbb{X}$, a positive constant can be found such that
\begin{equation}\notag
    O(||z(k+1)-x(k+1)||^2) \leqslant M||z(k+1)-x(k+1)||^2
\end{equation}
given that $x(k) = z(k)$, we have 
\begin{align}
    &J^*(z(k+1))-J^*(z(k))=J^*(z(k+1))-J^*(x(k))\notag\\
    &\leqslant J^*(x(k+1))-J^*(x(k)) + \frac{\partial J^*(x)}{\partial x}|_{x(k+1)}\cdot(z(k+1)-x(k+1))+M||z(k+1)-x(k+1)||^2\notag\\
    &\leqslant -\alpha(||x(k)-x_s||)+K_VL\delta+ML^2\delta^2\notag
\end{align}

If (\ref{realstabl}) holds, $\exists\rho_s$ such that for all $x(k)\in\Omega_{\rho_{\max}}$, we have $J^*(z(k+1))-J^*(z(k))\leqslant-\epsilon_s$, which implies that if $||x - x_s||_2\geqslant \rho_s$, the Lyapunov function keeps decreasing. Thus the system state will enter a region such that $||z - x_s||\leqslant\rho_s$ in finite steps.

Given that the definition of $\Omega_{\rho_{\min}}$, when the system state satisfies $||x - x_s||_2\leqslant \rho_s$, the state remains in $\Omega_{\rho_{\min}}$ at all times. Subsequently, the system state enters $\mathbb{X}_e$ in finite steps and remains within $\mathbb{X}_e$.
\end{proof}

\section{Simulation}\label{chap:4}

This section presents the closed-loop simulations that were performed with the complete nonlinear model (\ref{eq:z1}), and the parameters of (\ref{eq:z1}) are listed in Table \ref{tab:satpra}. The following results are conducted with the processor 5800X and the memory frequency 3600MHZ, with specific parameters of the optimization problem (\ref{mpc1:0}) detailed in Table \ref{zxjshs} for clarity.
\begin{table}[hbt!]
\caption{\label{tab:satpra} Satellite parameters}
\centering\small
\begin{tabular}{cc}
\hline
Parameters & Value  \\\hline
Satellite inertia matrix, $ I$ & $\mathrm{diag}(4.250,4.334,3.664)(\mathrm{kgm^2})$ \\
Axial wheel inertia, $I_w$ & $4\times 10^{-5}(\mathrm{kgm^2})$ \\ 
Axial wheel placement, $\Lambda$& $[0,1,0]^T$ \\ 
Thruster torque, $\tau_{\max}$ & $[0.0484, 0.0484, 0.0398]^T(\mathrm{Nm})$\\
Maximum wheel torque, $\tau_w$ & $0.0020(\mathrm{Nm})$\\
Maximum wheel velocity, $\omega_w$ &  $527(\mathrm{rad/s})$\\
\hline
\end{tabular}
\end{table}
\begin{table}[h]\small
\caption{Simulation parameters}
\label{zxjshs}
\vspace{0.5em}\centering
\begin{tabular}{ccc}
\hline
Parameters & Value   \\ \hline
state weight $\boldsymbol{Q}$ & $diag(500,500,500,1e-7,100,100,100)$  \\
input weight $\boldsymbol{R}$ & $diag(200,200,200,100)$  \\
control horizon $N_c$ & 24\\ 
sample period $T$ & 0.1($\mathrm{s}$) \\
$x_{\min}$ & $-[1,1,1,527,1,1,1]^T$  \\
$x_{\max}$ & $[1,1,1,527,1,1,1]^T$ \\
$u_{\min}$ & $-[0.0484,0.0484,0.0398,0.0020]^T$ \\
$u_{\max}$ & $[0.0484,0.0484,0.0398,0.0020]^T$ \\
initial $\boldsymbol{\omega}^b_{ob}$ & $[-0.05,0.15,-0.08]$($\mathrm{rad/s}$) \\
initial $\boldsymbol{\omega}_{w}$ & $300$($\mathrm{rad/s}$)  \\
initial $[\phi,\theta,\psi]$ & [-25,60,90]($\mathrm{deg}$) \\ 
expected $\boldsymbol{\omega}^b_{ob}$ & $[0,0,0]$($\mathrm{rad/s}$) \\
expected $\boldsymbol{\omega}_{w}$ & $0$($\mathrm{rad/s}$)  \\
expected $[\phi,\theta,\psi]$ & [0,0,0]($\mathrm{deg}$) \\
\hline
\end{tabular}
\end{table}

Linear MPC and lattice PWA approximation are used to perform the attitude control of the satellite and compared with the LQR method to highlight the superiority of the MPC controller for satisfying constraints. It is noted that traditional explicit MPC as listed in \cite{hegrenaes2005spacecraft} failed to provide the whole explicit control law.
\begin{figure}[htbp]
\centering
\includegraphics[width=0.45\textwidth]{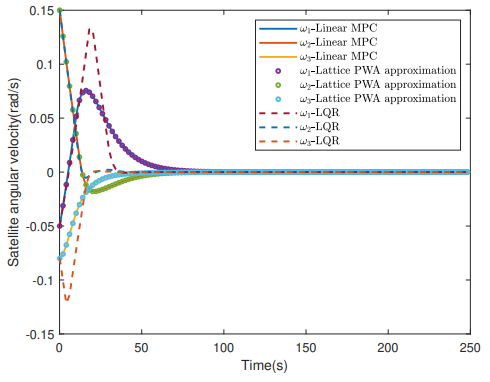}
\caption{Satellite angular velocity curve}
\label{wxjsd}
\end{figure}
\begin{figure}[htbp]
\centering
\includegraphics[width=0.45\textwidth]{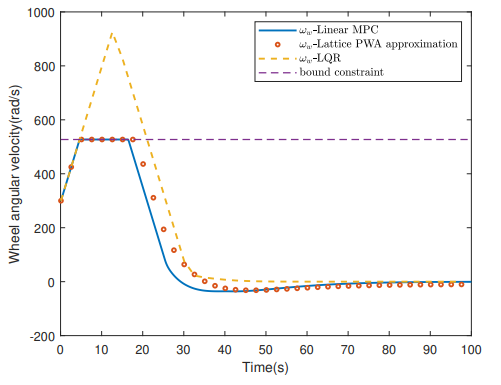}
\caption{Wheel angular velocity curve}
\label{f1}
\end{figure}
\begin{figure}[htbp]
\centering
\includegraphics[width=0.45\textwidth]{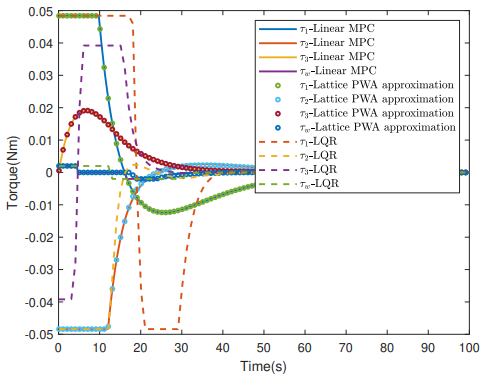}
\caption{Torque curve}
\label{lj}
\end{figure}
\begin{figure}[htbp]
\centering
\includegraphics[width=0.45\textwidth]{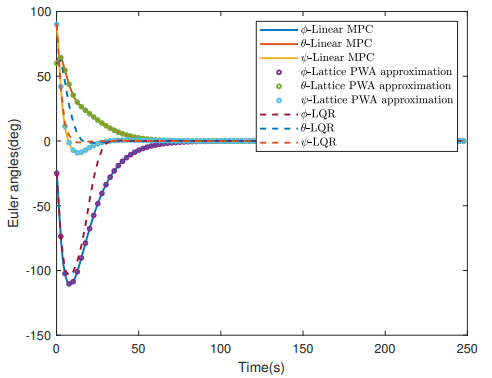}
\caption{Euler angle curve}
\label{olj}
\end{figure}

Fig. \ref{wxjsd}-\ref{olj} illustrates the simulation results of satellite attitude control, in which the method adopted in this paper has only a slight deviation from the results of linear MPC and does not affect the final stability of the system. The curves in Fig. \ref{wxjsd} are the angular velocity of the satellite coordinate system relative to the orbital coordinate system. It can be seen that all three controllers can achieve the ultimate elimination of relative rotation. Fig. \ref{f1} is the angular velocity curves of the flywheel, where the constraint on wheel angular velocity is not satisfied for the LQR method, while MPC and lattice PWA approximation can meet the constraints of wheel angular velocity. Fig. \ref{lj} is the external torque curves of the satellite, where all three controllers can meet the constraints on the input vector. Fig. \ref{olj} is the Euler angle curve, where the attitude of the satellite has reached the expected attitude under all controllers.

Table \ref{tab:simu_result} lists the online calculation time of each controller and the corresponding fuel consumption that is measured by the magnitude of the total impulse acting on the satellite. It can be seen that the online calculation time of the lattice PWA approximation is significantly shorter than that of the linear MPC. Besides, linear MPC and lattice PWA approximation can satisfy the state constraints of the satellite, while in the LQR method, the state constraints are not satisfied. Due to the advantages of rolling optimization, linear MPC and lattice PWA approximation consume less fuel than the LQR method.

\begin{table}[h]
\caption{Online calculation}
\label{tab:simu_result}
\vspace{0.5em}\centering
\begin{tabular}{cccc}
\hline
Methods & Constraint & Online time & Impulse \\\hline
Linear MPC & Y & 0.6e-2(s) & 1.27(Ns)\\
Lattice PWA & Y & 1.40e-4(s) & 1.27(Ns)\\ 
LQR& N & 5.1e-6(s) & 1.78(Ns)\\ 
\hline
\end{tabular}\par
\end{table}\par
	
\section{Conclusion}


The MPC problem based on satellite attitude control is a nonlinear optimization problem that is difficult to solve. In this study, Taylor expansion is used to linearize the attitude dynamics function with bounded error and transform the original optimization problem into an mpQP problem. Through the lattice PWA function, the explicit control law of the mpQP problem obtained by KKT conditions is approximated, and the online calculation is simplified. The stability of the original nonlinear systems subjected to lattice PWA approximation is proven considering the linearization error bound. The simulation results demonstrate that the proposed control scheme can increase the online calculation speed while ensuring control performance, and considerably outperforms the LQR method in terms of satisfying constraints and saving fuel.


\section*{Funding Sources}
This work was supported in part by the National Natural Science Foundation of China under Grant 62173113, and in part by the Science and Technology Innovation Committee of Shenzhen Municipality under Grant GXWD202311291016
52001, and in part by Natural Science Foundation of Guangdong Province of China under Grant 2022A1515011584.


\bibliography{sample}

\end{document}